\documentclass{article}%
\usepackage{amsmath}
\usepackage{amsfonts}
\usepackage{amssymb}
\usepackage{graphicx}%
\providecommand{\U}[1]{\protect\rule{.1in}{.1in}}
\newtheorem{theorem}{Theorem}

\newtheorem{algorithm}[theorem]{Algorithm}

\newtheorem{corollary}[theorem]{Corollary}

\newtheorem{lemma}[theorem]{Lemma}

\newenvironment{proof}[1][Proof]{\noindent\textbf{#1.} }{\ \rule{0.5em}{0.5em}}
\begin{document}

\author{Vadim E. Levit and David Tankus\\Department of Computer Science and Mathematics\\Ariel University Center of Samaria, ISRAEL\\\{levitv, davidta\}@ariel.ac.il}
\title{Well-Covered Graphs Without Cycles of Lengths 4, 5 and 6}
\date{}
\maketitle

\begin{abstract}
A graph $G$ is \textit{well-covered} if all its maximal independent sets are
of the same cardinality. Assume that a weight function $w$ is defined on its
vertices. Then $G$ is $w$\textit{-well-covered} if all maximal independent
sets are of the same weight. For every graph $G$, the set of weight functions
$w$ such that $G$ is $w$-well-covered is a \textit{vector space}. Given an
input graph $G$ without cycles of length $4$, $5$, and $6$, we characterize
polynomially the vector space of weight functions $w$ for which $G$ is $w$-well-covered.

Let $B$ be an induced complete bipartite subgraph of $G$ on vertex sets of
bipartition $B_{X}$ and $B_{Y}$. Assume that there exists an independent set
$S$ such that each of $S\cup B_{X}$ and $S\cup B_{Y}$ is a maximal independent
set of $G$. Then $B$ is a \textit{generating} subgraph of $G$, and it
\textit{produces} the restriction $w(B_{X})=w(B_{Y})$. It is known that for
every weight function $w$, if $G$ is $w$-well-covered, then the above
restriction is satisfied.

In the special case, where $B_{X}=\{x\}$ and $B_{Y}=\{y\}$, we say that $xy$
is a \textit{relating edge}. Recognizing relating edges and generating
subgraphs is an \textbf{NP-}complete problem. However, we provide a polynomial
algorithm for recognizing generating subgraphs of an input graph without
cycles of length $5$, $6$ and $7$. We also present a polynomial algorithm for
recognizing relating edges in an input graph without cycles of length $5$ and
$6$.

\end{abstract}

\section{Introduction}

Throughout this paper $G = (V,E)$ is a simple (i.e., a finite, undirected,
loopless and without multiple edges) graph with vertex set $V = V (G)$ and
edge set $E = E(G)$.

Cycles of $k$ vertices are denoted by $C_{k}$. When we say that $G$ does not
contain $C_{k}$ for some $k \geq3$, we mean that $G$ does not admit subgraphs
isomorphic to $C_{k}$. It is important to mention that these subgraphs are not
necessarily induced. Let $G(\widehat{C_{i_{1}}},..,\widehat{C_{i_{k}}})$ the
family of all graphs which do not contain $C_{i_{1}}$,...,$C_{i_{k}}$.

Let $u$ and $v$ be two vertices in $G$. The \textit{distance} between $u$ and
$v$, denoted $d(u,v)$, is the length of a shortest path between $u$ and $v$,
where the length of a path is the number of its edges. If $S$ is a non-empty
set of vertices, then the \textit{distance} between $u$ and $S$, denoted
$d(u,S)$, is defined by
\[
d(u,S)=\min\{d(u,s):s\in S\}.
\]
For every integer $i$, denote
\[
N_{i}(S)=\{x\in V:d(x,S)=i\},
\]
and
\[
N_{i}\left[  S\right]  =\{x\in V:d(x,S)\leq i\}.
\]

We abbreviate $N_{1}(S)$ and $N_{1}\left[  S\right]  $ to be $N(S)$ and
$N\left[  S\right]  $, respectively. If $S$ contains a single vertex, $v$,
then we abbreviate $N_{i}(\{v\})$, $N_{i}\left[  \{v\}\right]  $, $N(\{v\})$,
and $N\left[  \{v\}\right]  $ to be $N_{i}(v)$, $N_{i}\left[  v\right]  $,
$N(v)$, and $N\left[  v\right]  $, respectively. We denote by $G[S]$ the
subgraph of $G$ induced by $S$.

Let $G=(V,E)$ be a graph, and let $S$ and $T$ be two sets of vertices of $G$.
Then $S$ \textit{dominates} $T$ if $T\subseteq N\left[  S\right]  $. A set of
vertices $S$ is \textit{independent} if its elements are pairwise nonadjacent.
An independent set is \textit{maximal} if it is not a subset of another
independent set. The graph $G$ is \textit{well-covered} if all its maximal
independent sets are of the same cardinality. This concept was introduced by
Plummer in \cite{p:concepts}. The problem of finding a maximum cardinality
independent set is \textbf{NP-}complete. However, if the input is restricted
to well-covered graphs, then a maximum cardinality independent set can be
found polynomially using the \textit{greedy algorithm}.

Let $w:V\longrightarrow\mathbb{R}$ be a weight function defined on the
vertices of $G$. For every set $S\subseteq V$, define $w(S)=\Sigma_{s\in
S}w(s)$. Then $G$ is $w$-well-covered if all maximal independent sets of $G$
are of the same weight. The set of weight functions $w$ for which $G$ is
$w$-well-covered is a \textit{vector space} \cite{cer:degree}. We denote that
vector space $WCW(G)$ \cite{bnz:wcc4}. Clearly, $w\in WCW(G)$ if and only if
$G$ is $w$-well-covered. The \textit{dimension} of $WCW(G)$ is denoted by
$wcdim(G)$ \cite{bnz:wcc4}. More recent results about $wcdim$ can be found in
\cite{bv:wcdim1} and \cite{bkmuv:wcdim2}.

The recognition of well-covered graphs is known to be \textbf{co-NP}-complete.
This was proved independently in \cite{cs:note} and \cite{sknryn:compwc}. In
\cite{cst:structures} it is proven that the problem remains \textbf{co-NP}%
-complete even when the input is restricted to $K_{1,4}$-free graphs. However,
the problem is polynomially solvable for $K_{1,3}$-free graphs
\cite{tata:wck13f,tata:wck13fn}, for graphs with girth at least $5$
\cite{fhn:wcg5}, for graphs that contain neither $4$- nor $5$-cycles
\cite{fhn:wc45}, for graphs with a bounded maximal degree \cite{cer:degree},
or for chordal graphs \cite{ptv:chordal}. In \cite{lt:relating} there is a
polynomial characterization of well-covered graphs without cycles of length
$4$ and $6$.

Since recognizing well-covered graphs is \textbf{co-NP}-complete, finding
$WCW(G)$ is \textbf{co-NP}-complete as well. Recently, we developed a
polynomial time algorithm, which returns the vector space of weight functions
$w$ such that the input graph $G\in G(\widehat{C_{4}},\widehat{C_{5}%
},\widehat{C_{6}},\widehat{C_{7}})$ is $w$-well-covered.

\begin{theorem}
\label{wwcc4567} \cite{lt:wc4567} There exists a polynomial time algorithm,
which solves the following problem:\newline Input: A graph $G=(V,E)\in
G(\widehat{C_{4}},\widehat{C_{5}},\widehat{C_{6}},\widehat{C_{7}})$.\newline
Question: Find $WCW(G)$.
\end{theorem}

In order to prove Theorem \ref{wwcc4567}, the following notion has been
introduced in \cite{lt:wc4567}. Let $B$ be an induced complete bipartite
subgraph of $G$ on vertex sets of bipartition $B_{X}$ and $B_{Y}$. Assume that
there exists an independent set $S$ such that each of $S\cup B_{X}$ and $S\cup
B_{Y}$ is a maximal independent set of $G$. Then $B$ is a \textit{generating}
subgraph of $G$, and it \textit{produces} the restriction: $w(B_{X})=w(B_{Y}%
)$. Every weight function $w$ such that $G$ is $w$-well-covered must
\textit{satisfy} the restriction $w(B_{X})=w(B_{Y})$. The set $S$ is a
\textit{witness} that $B$ is generating. In the restricted case that the
generating subgraph $B$ is isomorphic to $K_{1,1}$, call its vertices $x$ and
$y$. In that case $xy$ is a \textit{relating} edge, and $w(x)=w(y)$ for every
weight function $w$ such that $G$ is $w$-well-covered. The decision problem
whether an edge in an input graph is relating is \textbf{NP-}complete
\cite{bnz:wcc4}. Therefore, recognizing generating subgraphs is \textbf{NP-}%
complete as well. However, recognizing relating edges can be done polynomially
if the input graph is restricted to $G(\widehat{C_{4}},\widehat{C_{6}})$
\cite{lt:relating}, and recognizing generating subgraphs is a polynomial
problem when the input graph is restricted to $G(\widehat{C_{4}}%
,\widehat{C_{6}},\widehat{C_{7}})$ \cite{lt:wc4567}.

In Section 2 we consider some general properties of $WCW(G)$. In Section 3 we
analyze the structure of $WCW(G)$ for graphs without cycles of length $5$. In
Section 4 we characterize polynomially relating edges in graphs without cycles
of length $5$ and $6$. In Section 5 we characterize polynomially generating
subgraphs in graphs without cycles of length $5$, $6$ and $7$. In Section 6 we
improve on Theorem \ref{wwcc4567} by presenting a polynomial algorithm which
solves the following problem:\newline\textit{Input:} A graph $G=(V,E)\in
G(\widehat{C_{4}},\widehat{C_{5}},\widehat{C_{6}})$.\newline\textit{Question:}
Find $WCW(G)$.

\section{The Vector Space $WCW(G)$}

\subsection{A Subspace of $WCW(G)$}

In this subsection we describe a procedure, which receives as its input a
graph $G=(V,E)$, and returns a vector space of weight functions
$w:V\longrightarrow\mathbb{R}$ such that $G$ is $w$-well-covered. This space
is a subspace of $WCW(G)$.

Recall that a vertex $v\in V$ is \textit{simplicial} if $N[v]$ is a complete graph.

\begin{theorem}
\label{pvs} Let $S$ be the set of all simplicial vertices in $G=(V,E)$, and
$A=\{a_{1},...,a_{k}\}$ be a maximal independent set of $G[S]$. Define a
weight function $w:V\longrightarrow\mathbb{R}$ as follows:

\begin{itemize}
\item for every $1\leq i\leq k$ choose an arbitrary value for $w(a_{i})$;

\item for every $v\in V\setminus A$ define $w(v)=w(N(v)\cap A)$.
\end{itemize}

Then $G$ is $w$-well-covered.
\end{theorem}

\begin{proof}
Let $X$ be a maximal independent set of $G$. Since $a_{i}$ is simplicial,
$\left\vert N\left[  a_{i}\right]  \cap X\right\vert =1$ for every $1\leq
i\leq k$. Therefore, $w(X)=\sum\limits_{1\leq i\leq k}w(a_{i})=w(A)$.
\end{proof}

Consequently, the collection of all weight functions obtained in accordance
with Theorem \ref{pvs} is a subspace of $WCW(G)$.

\subsection{The Dimension of $WCW(G)$}

In this subsection we present for each integer $k$, a family of graphs $\{
C_{m,k,r} \ | \ m \geq k \geq1, r \geq1 \}$ with the following properties:

\begin{enumerate}
\item $wcdim(C_{m,k,r})=k$ for each $1\leq k\leq m$ and for each $r\geq1$.

\item The size of $C_{m,k,r}$ limits to infinity when $m$ goes to infinity.
\end{enumerate}

Denote the vertices of the cycle $C_{m}$ by $v_{1},...,v_{m}$. The graph
$C_{m,k,r}$ is obtained from $C_{m}$ by adding new $k$ disjoint cliques,
$A_{1},...,A_{k}$, each of them is of size $r$. All vertices of $A_{i}$ are
adjacent to $v_{i}$, for each $1 \leq i \leq k$.

For every $1\leq i\leq k$, all vertices of $A_{i}$ are simplicial. A weight
function $w$ defined on the vertices of $C_{m,k,r}$ belongs to $WCW(C_{m,k,r}%
)$ if and only if it satisfies the following conditions.

\begin{enumerate}
\item $w(x)=w(y)$ for each $x,y \in\{v_{i}\} \cup A_{i} $ for each $1 \leq i
\leq k$.

\item $w(v_{i})=0$ for each $k+1 \leq i \leq m$.
\end{enumerate}

If $w\in WCW(C_{m,k,r})$ then the weight of every maximal independent set in
the graph is $\Sigma_{1\leq i\leq k}w(v_{i})$, and $wcdim(C_{m,k,r})=k$.

\begin{figure}[h]
\setlength{\unitlength}{1.0cm} \begin{picture}(20,8)\thicklines
\put(6.5,0.5){\circle*{0.2}}
\put(6.5,5.5){\circle*{0.2}}
\put(1,3){\circle*{0.2}}
\put(3.5,0.5){\circle*{0.2}}
\put(3.5,5.5){\circle*{0.2}}
\put(9,3){\circle*{0.2}}
\put(3.5,5.2){\makebox(0,0){$v_{1}$}}
\put(6.5,5.2){\makebox(0,0){$v_{2}$}}
\put(9,2.7){\makebox(0,0){$v_{3}$}}
\put(6.5,0.2){\makebox(0,0){$v_{4}$}}
\put(3.5,0.2){\makebox(0,0){$v_{5}$}}
\put(1,2.7){\makebox(0,0){$v_{6}$}}
\put(3.5,5.5){\line(1,0){3}}
\put(3.5,0.5){\line(1,0){3}}
\put(1,3){\line(1,1){2.5}}
\put(1,3){\line(1,-1){2.5}}
\put(9,3){\line(-1,1){2.5}}
\put(9,3){\line(-1,-1){2.5}}
%K4 of v1
\put(4,6.5){\circle*{0.2}}
\put(3,6.5){\circle*{0.2}}
\put(4,7.5){\circle*{0.2}}
\put(3,7.5){\circle*{0.2}}
\put(3.5,5.5){\line(1,2){0.5}}
\put(3.5,5.5){\line(-1,2){0.5}}
\put(3.5,5.5){\line(1,4){0.5}}
\put(3.5,5.5){\line(-1,4){0.5}}
\put(3,6.5){\line(0,1){1}}
\put(3,6.5){\line(1,0){1}}
\put(3,6.5){\line(1,1){1}}
\put(4,7.5){\line(0,-1){1}}
\put(4,7.5){\line(-1,0){1}}
\put(4,6.5){\line(-1,1){1}}
\put(2.7,7){\makebox(0,0){$A_{1}$}}
%K4 of v2
\put(7,6.5){\circle*{0.2}}
\put(6,6.5){\circle*{0.2}}
\put(7,7.5){\circle*{0.2}}
\put(6,7.5){\circle*{0.2}}
\put(6.5,5.5){\line(1,2){0.5}}
\put(6.5,5.5){\line(-1,2){0.5}}
\put(6.5,5.5){\line(1,4){0.5}}
\put(6.5,5.5){\line(-1,4){0.5}}
\put(6,6.5){\line(0,1){1}}
\put(6,6.5){\line(1,0){1}}
\put(6,6.5){\line(1,1){1}}
\put(7,7.5){\line(0,-1){1}}
\put(7,7.5){\line(-1,0){1}}
\put(7,6.5){\line(-1,1){1}}
\put(5.7,7){\makebox(0,0){$A_{2}$}}
%K4 of v3
\put(9.5,4){\circle*{0.2}}
\put(8.5,4){\circle*{0.2}}
\put(9.5,5){\circle*{0.2}}
\put(8.5,5){\circle*{0.2}}
\put(9,3){\line(1,2){0.5}}
\put(9,3){\line(-1,2){0.5}}
\put(9,3){\line(1,4){0.5}}
\put(9,3){\line(-1,4){0.5}}
\put(8.5,4){\line(0,1){1}}
\put(8.5,4){\line(1,0){1}}
\put(8.5,4){\line(1,1){1}}
\put(9.5,5){\line(0,-1){1}}
\put(9.5,5){\line(-1,0){1}}
\put(9.5,4){\line(-1,1){1}}
\put(8.2,4.5){\makebox(0,0){$A_{3}$}}
\end{picture}
\caption{The graph $C_{6,3,4}$.}%
\label{c634}%
\end{figure}
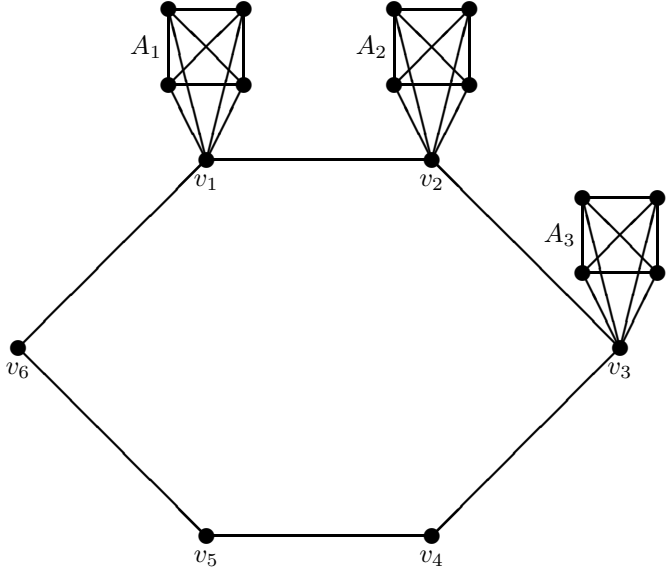

\section{Graphs Without Cycles of Length $5$}

Let $G = (V,E) \in G(\widehat{C_{5}})$, and let $w:V \longrightarrow
\mathbb{R}$. In this section we find a necessary condition that $G$ is $w$-well-covered.

Define $L(G)$ to be the set of all vertices $v \in V$ such that one of the
following holds:

\begin{enumerate}
\item $d(v)=1$.

\item $d(v)=2$ and $v$ is on a triangle.
\end{enumerate}

For every $v \in V$ define $D(v) = N(v) \setminus N(N_{2}(v))$, and let $M(v)$
be a maximal independent set of $D(v)$.

The fact that $G \in G(\widehat{C_{5}})$ implies that for every $v \in V$, the
subgraph induced by $D(v)$ can not contain a path of length 3. Therefore,
every connected component of $D(v)$ is either a $K_{3}$ or a star. ($K_{1}$
and $K_{2}$ are restricted cases of a star.)

\begin{lemma}
\label{dc5} Let $G=(V,E)\in G(\widehat{C_{5}})$, and let $v\in V$. Then every
maximal independent set of $N_{2}(v)$ dominates $N(v)\cap N(N_{2}(v))$.
\end{lemma}

\begin{proof}
Let $v\in V$, and let $T$ be a maximal independent set of $N_{2}(v)$. Assume
on the contrary that $T$ does not dominate $N(v)\cap N(N_{2}(v))$. Let
$u\in(N(v)\cap N(N_{2}(v)))\setminus N(T)$, and let $u^{\prime}\in N(u)\cap
N_{2}(v)$. Clearly, $u^{\prime}\not \in T$ but $u^{\prime}$ is adjacent to a
vertex $t\in T$. The fact that $t\in N_{2}(v)$ implies that there exists a
vertex $x\in N(t)\cap N(v)$. Hence, $(v,u,u^{\prime},t,x)$ is a cycle of
length 5, which is a contradiction. Therefore, every maximal independent set
of $N_{2}(v)$ dominates $N(v)\cap N(N_{2}(v))$.
\end{proof}

\begin{corollary}
\label{dc5c} Let $G=(V,E)\in G(\widehat{C_{5}})$, and let $v\in V$. If
$D(v)=\emptyset$ then every maximal independent set of $N_{2}(v)$ dominates
$N(v)$.
\end{corollary}

\begin{proof}
If $D(v)=\emptyset$, then $N(v)\cap N(N_{2}(v))=N(v)$, and, consequently, by
Lemma \ref{dc5}, every maximal independent set of $N_{2}(v)$ dominates $N(v)$.
\end{proof}

\begin{theorem}
\label{wnlc5} Assume that $G=(V,E)\in G(\widehat{C_{5}})$ is $w$-well-covered
for some weight function $w:V\longrightarrow\mathbb{R}$. Then $D(v)\neq
\emptyset$ implies $w(v)=w(M(v))$ for every $v\in V\setminus L(G)$.
\end{theorem}

\begin{proof}
Let $v\in V\setminus L(G)$, let $T$ be a maximal independent set of $N_{2}%
(v)$, and let $S$ be a maximal independent set of $G\setminus N\left[
v\right]  $, which contains $T$. Then $S\cup\{v\}$ and $S\cup M(v)$ are two
maximal independent sets of $G$. The fact that $G$ is $w$-well-covered implies
that $w(S\cup\{v\})=w(S\cup M(v))$. Therefore, $w(v)=w(M(v))$.
\end{proof}

\section{Relating Edges in Graphs Without Cycles of Length 5 and 6}

In this section we prove that recognizing relating edges in an input graph,
which does not contain cycles of length $5$ and $6$, is a polynomial problem.

\begin{theorem}
\label{c456related} Let $G=(V,E)\in G(\widehat{C_{5}},\widehat{C_{6}})$ and
let $xy\in E$. Then $xy$ is relating if and only if $N_{2}(\{x,y\})$ dominates
$N(x)\bigtriangleup N(y)$.
\end{theorem}

\begin{proof}
Assume that $N_{2}(\{x,y\})$ dominates $N(x)\bigtriangleup N(y)$. The
following algorithm returns a witness that $xy$ is relating.

Construct a set $S_{x}$ as follows. For every vertex $x^{\prime}\in N(x)$ add
to $S_{x}$ a vertex $x^{\prime\prime}\in N(x^{\prime})\cap N_{2}(x)$. The set
$S_{x}$ is independent, because if $x_{1}^{\prime\prime}$ and $x_{2}%
^{\prime\prime}$ were two adjacent vertices in $S_{x}$, then there existed two
distinct vertices $x_{1}^{\prime}\in N(x)\cap N(x_{1}^{\prime\prime})$ and
$x_{2}^{\prime}\in N(x)\cap N(x_{2}^{\prime\prime})$. Hence, $(x,x_{1}%
^{\prime},x_{1}^{\prime\prime},x_{2}^{\prime\prime},x_{2}^{\prime})$ was a
cycle of length $5$.

Construct similarly an independent set $S_{y}$ by choosing a vertex
$y^{\prime\prime}\in N(y^{\prime})\cap N_{2}(y)$ for every $y^{\prime}\in
N(y)$. The fact that $G$ does not contain cycles of length $6$ implies that
$S_{x}\cup S_{y}$ is independent too. Let $S$ be a maximal independent set of
$G\setminus N\left[  \{x,y\}\right]  $ which contains $S_{x}\cup S_{y}$. Then
$S\cup\{x\}$ and $S\cup\{y\}$ are maximal independent sets of $G$. Therefore,
$xy$ is related, and $S$ is the witness.

Assume $xy$ is relating. Let $S$ be a witness that $xy$ is relating. Then
$S\cap N_{2}(\{x,y\})$ dominates $N(x)\bigtriangleup N(y)$.
\end{proof}

\section{Generating \ Subgraphs \ in \ Graphs \ Without \ Cycles of Lengths
$5$, $6$ and $7$}

In this section we prove that recognizing generating subgraphs in an input
graph, which does not contain cycles of lengths $5,6$ and $7$, is a polynomial problem.

\begin{theorem}
\label{generating567} Let $G\in G(\widehat{C_{5}},\widehat{C_{6}%
},\widehat{C_{7}})$, and let $B$ be an induced complete bipartite subgraph of
$G$ on vertex sets of bipartition $B_{X}$ and $B_{Y}$. Then $B$ is generating
if and only if $N_{2}(B)$ dominates $N(B_{X})\bigtriangleup N(B_{Y})$.
\end{theorem}

\begin{proof}
Assume that $B$ is generating. Let $S$ be a witness of $B$. Then $S\cap
N_{2}(B)$ dominates $N(B_{X})\bigtriangleup N(B_{Y})$, therefore $N_{2}(B)$
dominates $N(B_{X})\bigtriangleup N(B_{Y})$.

Suppose that $N_{2}(B)$ dominates $N(B_{X})\bigtriangleup N(B_{Y})$. Let
$S_{X}$ be a maximal independent set of $N_{2}(B_{X})\cap N_{3}(B_{Y})$, and
let $S_{Y}$ be a maximal independent set of $N_{2}(B_{Y})\cap N_{3}(B_{X})$.
The fact that $G$ does not contain cycles of length $6$ implies that
$S=S_{X}\cup S_{Y}$ is independent. The fact that $G$ does not contain cycles
of length $5$ implies that there are no edges between $S_{X}$ and
$N(B_{Y})\cap N_{2}(B_{X})$. Similarly, there are no edges between $S_{Y}$ and
$N(B_{X})\cap N_{2}(B_{Y})$.

Assume on the contrary that there exists a vertex $x^{\prime}\in N(B_{X}) \cap
N_{2}(B_{Y})$ which is not dominated by $S$. Clearly, $x^{\prime}$ is adjacent
to a vertex $x^{\prime\prime}\in N_{2}(B_{X}) \cap N_{3}(B_{Y})$. Hence,
$x^{\prime\prime}$ is a neighbor of a vertex $v \in S_{X}$. Clearly, $v$ is
adjacent to a vertex $w \in N(B_{X}) \cap N_{2}(B_{Y})$.

Let $x_{1}$ be a neighbor of $x^{\prime}$ in $B_{X}$, and let $x_{2}$ be a
neighbor of $w$ in $B_{X}$. If $x_{1}=x_{2}$ then $(x_{1},x^{\prime}%
,x^{\prime\prime},v,w)$ is a cycle of length $5$. Otherwise, let $y$ be any
vertex of $B_{Y}$. Then $(x_{1},x^{\prime},x^{\prime\prime},v,w,x_{2},y)$ is a
cycle of length $7$. In both cases we obtained a contradiction. Therefore $S$
dominates $N(B_{X})\cap N_{2}(B_{Y})$. Similarly, $S$ dominates $N(B_{Y})\cap
N_{2}(B_{X})$. Let $S^{\ast}$ be any maximal independent set of $G\setminus
N[B]$ which contains $S$. Then $S^{\ast}$ is a witness that $B$ is generating.
\end{proof}

\section{The Vector \ Space \ of \ Well-Covered \ Graphs Without Cycles of
Lengths $4$, $5$, and $6$}

In this section $G=(V,E)\in G(\widehat{C_{4}},\widehat{C_{5}},\widehat{C_{6}%
})$. Therefore, $L(G)$ is the set of all simplicial vertices of $G$. For each
$v\in V$ every connected component of $N(v)$ is either a $K_{1}$ or a $K_{2}$.
Also, every connected component of $D(v)$ is either a $K_{1}$ or a $K_{2}$.

\subsection{A polynomial characterization of $WCW(G)$.}

\begin{theorem}
\label{wwcc456} Let $G\in G(\widehat{C_{4}},\widehat{C_{5}},\widehat{C_{6}})$.
There exists a polynomial time algorithm which finds $WCW(G)$.
\end{theorem}

The proof of Theorem \ref{wwcc456} is based on the polynomial characterization
of well-covered graphs without cycles of length $4$ and $5$, found by Finbow,
Hartnell and Nowakowski.

\begin{theorem}
\cite{fhn:wc45} \label{wcc45} Let $H=(V,E)\in G(\widehat{C_{4}},\widehat{C_{5}%
})$. Then $H$ is well-covered if and only if one of the following conditions holds.

\begin{enumerate}
\item There exists a set $\{v_{1},...,v_{k}\} \subseteq V$ of \ simplicial
\ vertices \ such \ that $|N[v_{i}]| \leq3$ for every $1 \leq i \leq k$, and
$\{N[v_{i}] | 1 \leq i \leq k\}$ is a partition of $V$.

\item $H$ is isomorphic to $C_{7}$ or to $T_{10}$.
\end{enumerate}
\end{theorem}

\begin{figure}[h]
\setlength{\unitlength}{1.0cm} \begin{picture}(20,5)\thicklines
\multiput(1,0.5)(0,1){3}{\circle*{0.2}}
\multiput(3,0.5)(0,1){3}{\circle*{0.2}}
\multiput(2,1.5)(0,1){2}{\circle*{0.2}}
\put(2,1){\circle*{0.2}}
\put(2,3){\circle*{0.2}}
\put(1,0.5){\line(0,1){2}}
\put(1,0.5){\line(1,0){2}}
\put(1,0.5){\line(2,1){1}}
\put(3,0.5){\line(0,1){2}}
\put(2,1){\line(0,1){2}}
\put(2,1){\line(2,-1){1}}
\put(1,2.5){\line(2,1){1}}
\put(2,3){\line(2,-1){1}}
\multiput(8,0.5)(1,0){3}{\circle*{0.2}}
\multiput(7,1.5)(0,1){3}{\circle*{0.2}}
\multiput(8,4.5)(1,0){3}{\circle*{0.2}}
\multiput(11,1.5)(0,1){3}{\circle*{0.2}}
\put(8,0.5){\line(1,0){2}}
\put(7,1.5){\line(0,1){2}}
\put(8,4.5){\line(1,0){2}}
\put(11,1.5){\line(0,1){2}}
\put(7,1.5){\line(1,-1){1}}
\put(7,3.5){\line(1,1){1}}
\put(11,1.5){\line(-1,-1){1}}
\put(11,3.5){\line(-1,1){1}}
\put(7,2.5){\line(1,0){4}}
\put(9,0.5){\line(0,1){4}}
\put(8,4.8){\makebox(0,0){$v_{12}$}}
\put(9,4.8){\makebox(0,0){$v_{1}$}}
\put(10,4.8){\makebox(0,0){$v_{2}$}}
\put(11.3,3.5){\makebox(0,0){$v_{3}$}}
\put(11.3,2.5){\makebox(0,0){$v_{4}$}}
\put(11.3,1.5){\makebox(0,0){$v_{5}$}}
\put(10,0.2){\makebox(0,0){$v_{6}$}}
\put(9,0.2){\makebox(0,0){$v_{7}$}}
\put(8,0.2){\makebox(0,0){$v_{8}$}}
\put(6.6,1.5){\makebox(0,0){$v_{9}$}}
\put(6.6,2.5){\makebox(0,0){$v_{10}$}}
\put(6.6,3.5){\makebox(0,0){$v_{11}$}}
\end{picture}\caption{The graphs $T_{10}$ (left) and $D_{12}$ (right).}%
\label{T10}%
\end{figure}
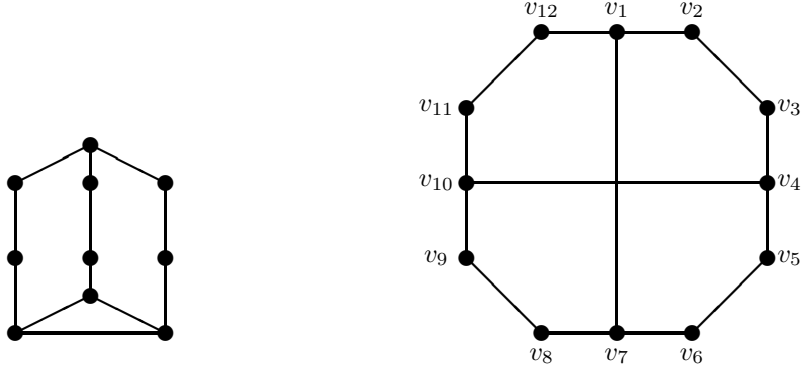

The following lemmas together with Theorem \ref{wcc45} imply Theorem
\ref{wwcc456}.

\begin{lemma}
\label{lc12} Let $w:V\longrightarrow\mathbb{R}$ be a weight function defined
on $G$, and assume that $G$ is $w$-well-covered. Let $v\in V\setminus L(G)$.
If there exists a vertex $u\in N\left[  v\right]  $ such that $D(u)\neq
\emptyset$ then $w(v)=w(M(v))$.
\end{lemma}

\begin{proof}
If $D(v)\neq\emptyset$ then according to Theorem \ref{wnlc5}, $w(v)=w(M(v))$,
and the lemma holds.

Suppose $D(v)=\emptyset$ and there exists a vertex $u\in N(v)$ such that
$D(u)\neq\emptyset$. Let $T$ be a maximal independent set of $(N_{2}(v)\cap
N_{3}(u))\cup(N_{2}(u)\cap N_{3}(v))$, let $S_{1}$ be a maximal independent
set of $G$ such that $S_{1}\supseteq T\cup\{u\}$, and let $S_{2}%
=(S_{1}\setminus\{u\})\cup M(u)\cup\{v\}$. Clearly, $w(S_{1})=w(S_{2})$,
therefore $w(u)=w(M(u)\cup\{v\})$. However, by Theorem, \ref{wnlc5}
$w(u)=w(M(u))$. Hence $w(v)=0=w(M(v))$.
\end{proof}

\begin{lemma}
\label{ifflnephi} Assume that $L(G) \neq\emptyset$. Let $w:V\longrightarrow
\mathbb{R}$ be a weight function defined on the vertices of $G$. Then $G$ is
$w$-well-covered if and only if $w(v)=w(M(v))$ for every $v \in V \setminus
L(G)$.
\end{lemma}

\begin{proof}
\textit{If part:} Assume that $G$ is $w$-well-covered. Let $v \in V \setminus
L(G)$. It is enough to prove that $w(v)=w(M(v))$.

If there exists a vertex $u\in N\left[  v\right]  $ such that $D(u)\neq
\emptyset$ then $w(v)=w(M(v))$ by Lemma \ref{lc12}.

Suppose that $D(u)=\emptyset$ for every $u\in N\left[  v\right]  $. Let $H$ be
the subgraph of $G$ induced by $\{v\in V:D(v)=\emptyset\}$. The fact that
$L(G)\neq\emptyset$ implies that $H\neq G$. Let $C$ be the connected component
of $H$ which contains $v$. By Theorem \ref{c456related} all edges of $C$ are
relating. Therefore, all vertices of $C$ are of the same weight. There exists
a vertex $x\in C$ which is adjacent to a vertex $y\notin C$. Clearly,
$D(y)\neq\emptyset$. Lemma \ref{lc12} implies that $w(x)=w(M(x))=0$. Hence,
$w(z)=0=w(M(z))$ for every vertex $z\in C$.

\textit{Only if part:} Assume that $w(v)=w(M(v))$ for every $v \in V \setminus
L(G)$, and $L(G) \neq\emptyset$. Since $L(G)$ is the set of simplicial
vertices in the graph, Theorem \ref{pvs} implies that $w \in VS(G)$.
\end{proof}

\begin{corollary}
\label{c456wc} Assume that $L(G) \neq\emptyset$. Then $G$ is well-covered if
and only if $D(v)$ is a copy of $K_{1}$ or $K_{2}$, for every $v \in V
\setminus L(G)$.
\end{corollary}

\begin{lemma}
\label{lphirelating} If $L(G)=\emptyset$ then all edges of $G$ are relating.
\end{lemma}

\begin{proof}
$L(G)=\emptyset$ implies $D(v)=\emptyset$ for every $v\in V$. Therefore,
$N_{2}(\{x,y\})$ dominates $N(x)\bigtriangleup N(y)$ for every edge $xy$.
Hence, by Theorem \ref{c456related} all edges of $G$ are relating.
\end{proof}

\begin{lemma}
\label{lphi} Assume that $L(G)=\emptyset$. Then the following holds.

\begin{enumerate}
\item If $G$ is isomorphic to either $C_{7}$ or $T_{10}$ Then $w\in WCW(G)$ if
and only if there exists $k\in\mathbb{R}$ such that $w\equiv k$.

\item Otherwise, $WCW(G)$ contains only the zero function.
\end{enumerate}
\end{lemma}

\begin{proof}
By Lemma \ref{lphirelating}, all edges of $G$ are relating. Therefore, if
$w\in WCW(G)$, then all vertices of $G$ are of the same weight. Hence, it
should be decided which of the following two cases holds.

\begin{enumerate}
\item $G$ is well-covered. Hence, $w\in WCW(G)$ if and only if there exists
$k\in\mathbb{R}$ such that $w\equiv k$. In this case $wcdim(G)=1$.

\item $G$ is not well-covered. Hence, $w\in WCW(G)$ if and only if $w\equiv0$.
In this case $wcdim(G)=0$.
\end{enumerate}

Since $L(G)=\emptyset$, there are no simplicial vertices in $G$. Consequently,
the first condition of Theorem \ref{wcc45} does not hold. By Theorem
\ref{wcc45} the graph is well-covered if and only if it is isomorphic to
$C_{7}$ or to $T_{10}$.
\end{proof}

An example of the above is the graph $D_{12}$. Clearly, $L(D_{12})=\Phi$, and
the graph does not contain simplicial vertices. All edges in the graph are
relating. The graph is not well-covered because $\{v_{3},v_{6},v_{9},v_{12}\}$
and $\{v_{3},v_{5},v_{7},v_{9},v_{12}\}$ are two maximal independent sets of
with different cardinalities. Therefore $WCW(D_{12})$ contains only the zero function.

\subsection{The Algorithm and its Complexity}

The following algorithm receives as its input a graph $G=(V,E)\in
G(\widehat{C_{4}},\widehat{C_{5}},\widehat{C_{6}})$, and finds $WCW(G)$. All
elements of $WCW(G)$ can be obtained by this algorithm.

\begin{algorithm}
Vector Space
\end{algorithm}

\begin{enumerate}
\item Find $L(G)$.

\item \textbf{If} $G$ is isomorphic to $C_{7}$ or to $T_{10}$

\begin{enumerate}
\item Assign an arbitrary value for $k$.

\item For each $v\in V$ denote $w(v)=k$
\end{enumerate}

\item \textbf{Else}

\begin{enumerate}
\item Find a maximal independent set $S$ of $L(G)$, and assign arbitrary
weights to the elements of $S$.

\item For each vertex $l\in L(G)\setminus S$ denote $w(l)=w(N(l)\cap S)$.

\item For each $v\in V\setminus L(G)$

\begin{enumerate}
\item Find $D(v)$.

\item Construct a maximal independent set $M(v)$ of $D(v)$.

\item Denote $w(v)=w(M(v))$
\end{enumerate}
\end{enumerate}
\end{enumerate}

\textbf{Correctness of the algorithm.} If the condition of Step 2 holds, then
by Theorem \ref{wcc45} and Lemma \ref{lphirelating}, the algorithm returns
$WCW(G)$.

Assume that the condition of Step 2 does not hold. Denote the elements of the
set $S$ found in Step 3a by $S=\{s_{1},...,s_{|S|}\}$. Then $wcdim(G)=|S|$,
and $w(s_{1}),...,w(s_{|S|})$ are the free variables of the vector space.
Every connected component of $L(G)$ contains at most $2$ vertices. According
to Step 3b, if $l_{1}$ and $l_{2}$ are two vertices of the same connected
component of $L(G)$, then $w(l_{1})=w(l_{2})$. Therefore, if $x\in V\setminus
L(G)$, and $T_{1}$, $T_{2}$ are two maximal independent sets of $G[N(x) \cap
L(G)]$, then $w(T_{1})=w(T_{2})$.

In Step 3c, $D(v)\subseteq L(G)$, for every vertex $v\in V\setminus L(G)$. The
set $M(v)$ can be constructed in more than one possible way. However,
$w(M(v))$ is uniquely defined by Step 3b. If $L(G)\neq\emptyset$, by Lemma
\ref{ifflnephi}, the algorithm returns $WCW(G)$. If $L(G)=\emptyset$, by Lemma
\ref{lphi}, the algorithm returns $WCW(G)$.

\textbf{Complexity analysis.} Steps 1 and 2 run in $O(\left\vert V\right\vert
)$ time. Steps 3a and 3b can be implemented in $O(|E|)$ time. Step 3c is a
loop with $\left\vert V\right\vert $ iterations. Each iteration can be
implemented in $O(\left\vert E\right\vert )$ time. Therefore, the total
complexity of Step 3c is $O(\left\vert V\right\vert \left\vert E\right\vert
)$, which is the total complexity of the whole algorithm as well.

\section{Open Question}

In this paper we presented a polynomial algorithm whose input is a graph in
$G(\widehat{C_{4}},\widehat{C_{5}},\widehat{C_{6}})$, and its output is
$WCW(G)$. On the other hand, there is a polynomial characterization of
well-covered graphs without cycles of lengths $4$ and $5$ \cite{fhn:wc45}.
Thus it is a natural step in learning $w$-well-covered graphs to ask whether
the following problem is polynomially solvable.\newline\textit{Input:} A graph
$G=(V,E)\in G(\widehat{C_{4}},\widehat{C_{5}})$.\newline\textit{Output:}
$WCW(G)$.


\begin{thebibliography}{99}                                                                                               %


\bibitem {bv:wcdim1}I. Birnbaum, O. Vegaj, \emph{Various Results on the
Well-Covered Dimension of a Graph}, arXiv:1003.3968v2 [math.CO] 2011.

\bibitem {bkmuv:wcdim2}I. Birnbaum, M. Kuneli, R. McDonald, K. Urabe, O.
Vegaj, \emph{The Well-Covered Dimension of Products of Graphs},
arXiv:1003.3968v3 [math.CO] 2012.

\bibitem {bnz:wcc4}J. I. Brown, R. J. Nowakowski, I. E. Zverovich, \emph{The
structure of well-covered graphs with no cycles of length 4}, Discrete
Mathematics \textbf{307} (2007) 2235-2245.

\bibitem {cer:degree}Y. Caro, N. Ellingham, G. F. Ramey, \emph{Local structure
when all maximal independent sets have equal weight}, SIAM Journal on Discrete
Mathematics \textbf{11} (1998) 644-654.

\bibitem {cst:structures}Y. Caro, A. Seb\H{o}, M. Tarsi, \emph{Recognizing
greedy structures}, Journal of Algorithms \textbf{20} (1996) 137-156.

\bibitem {cs:note}V. Chvatal, P. J. Slater, \emph{A note on well-covered
graphs}, Quo vadis, Graph Theory Annals of Discrete Mathematics \textbf{55},
North Holland, Amsterdam, 1993, 179-182.

\bibitem {fhn:wcg5}A. Finbow, B. Hartnell, R. Nowakowski, \emph{A
characterization of well-covered graphs of girth 5 or greater}, Journal of
Combinatorial Theory B \textbf{57} (1993) 44-68.

\bibitem {fhn:wc45}A. Finbow, B. Hartnell, R. Nowakowski, \emph{A
characterization of well-covered graphs that contain neither 4- nor 5-cycles},
Journal of Graph Theory \textbf{18} (1994) 713-721.

\bibitem {lt:relating}V. Levit, D. Tankus, \emph{On relating edges in
well-covered graphs without cycles of length 4 and 6}, Graph Theory,
Computational Intelligence and Thought: Essays Dedicated to Martin Charles
Golumbic on the Occasion of His 60th Birthday, Lecture Notes in Computer
Science \textbf{5420} (2009) 144-147.

\bibitem {lt:wc4567}V. Levit, D. Tankus \emph{Weighted well-covered graphs
without $C_{4}$, $C_{5}$, $C_{6}$, $C_{7}$}, Discrete Applied Mathematics
\textbf{159} (2011) 354-359.

\bibitem {p:concepts}M.D. Plummer \emph{Some covering concepts in graphs},
Journal of Combinatorial Theory \textbf{8} (1970) 91-98.

\bibitem {ptv:chordal}E. Prisner, J. Topp and P. D. Vestergaard,
\emph{Well-covered simplicial, chordal and circular arc graphs}, Journal of
Graph Theory \textbf{21} (1996), 113--119.

\bibitem {sknryn:compwc}R. S. Sankaranarayana, L. K. Stewart, \emph{Complexity
results for well-covered graphs}, Networks \textbf{22} (1992), 247--262.

\bibitem {tata:wck13f}D. Tankus, M. Tarsi, \emph{Well-covered claw-free
graphs}, Journal of Combinatorial Theory B \textbf{66} (1996) 293-302.

\bibitem {tata:wck13fn}D. Tankus, M. Tarsi, \emph{The structure of
well-covered graphs and the complexity of their recognition problems}, Journal
of Combinatorial Theory B \textbf{69} (1997) 230-233.
\end{thebibliography}
\end{document}